\documentclass[conference]{IEEEtran}
\usepackage{cite}
\usepackage{enumitem}

\usepackage{amsfonts, amsmath, amssymb, amsthm}
\usepackage{array}
\usepackage{tikz, color, soul}
\usepackage{pgfplots}
\usepackage[hidelinks]{hyperref} 
\usetikzlibrary{plotmarks}
\usepackage{mathtools}
\usepackage{flushend}
\usepackage{comment}

\def\orcidID#1{\unskip$^{[#1]}$} 
\usepackage[ruled,linesnumbered]{algorithm2e}
\usetikzlibrary{shapes.geometric,arrows,positioning}
\tikzstyle{decision} = [diamond, draw, fill=blue!20, 
text width=5em, text badly centered, node distance=4cm, inner sep=0pt]
\tikzstyle{block} = [rectangle, draw, fill=blue!20,  text centered, rounded corners, minimum height=4em]
\tikzstyle{line} = [draw, -latex']
\tikzstyle{cloud} = [draw, ellipse,fill=red!20, node distance=6.6cm,
minimum height=1em]
\tikzstyle{algorithm} = [rectangle, draw, fill=green!20,  text centered, rounded corners, minimum height=4em, minimum width =6em]
\tikzstyle{initialization} = [rectangle, draw,   text centered, minimum height=4em, minimum width =6em]

\allowdisplaybreaks
\usetikzlibrary{spy}
\usetikzlibrary{backgrounds}
\usetikzlibrary{decorations}
\usetikzlibrary{patterns}
\usetikzlibrary{arrows,matrix,positioning,calc}

\newtheorem{theorem}{Theorem}
\newtheorem*{theorem*}{Theorem}

\newtheorem{lemma}{Lemma}
\newtheorem{corollary}{Corollary}
\theoremstyle{definition}

\theoremstyle{remark}
\newtheorem*{remark}{Remark}
\newtheorem*{example}{Example}
\usepackage{thm-restate}

\usepackage{xcolor}
\newif\ifcomment
\commenttrue
\renewcommand{\epsilon}{\varepsilon}

\tikzstyle{block}=[draw, rectangle, minimum height=1cm, text width=1.5cm, text centered, draw=darkgray, font=\small]
\tikzstyle{block_medium}=[draw, rectangle, minimum height=1.5cm, text width=2cm, text centered, draw=darkgray, font=\small]
\tikzstyle{block_large}=[draw, rectangle, minimum height=2cm, text width=2cm, text centered, draw=darkgray, font=\small]
\tikzstyle{line} = [draw, -latex]

\usepackage{mathtools}

\pgfplotsset{compat=1.14}
\begin{document}
	\title{Signature Codes for a Noisy Adder Multiple Access Channel
	}
	\author{%
		\IEEEauthorblockN{Gökberk~Erdoğan \orcidID{0000-0002-6282-0894}, Georg~Maringer \orcidID{0000-0002-2868-5131}, Nikita~Polyanskii \orcidID{0000-0003-3735-5705}}\\
		\thanks{
			G.~Erdoğan and G.~Maringer are with the Institute for Communications Engineering, Technical University of Munich, Germany. N.~Polyanskii is with the IOTA Foundation, Germany.
			
			Georg Maringer's work was supported by the German Research Foundation (Deutsche Forschungsgemeinschaft, DFG) under Grant No.~WA3907/4-1.
			
				Nikita Polyanskii's work was partially done
 with the Technical University of Munich and the Skolkovo Institute of Science and Technology in 2021. Nikita Polyanskii's research was supported by the German Research Foundation (Deutsche Forschungsgemeinschaft, DFG) under Grant No. WA3907/1-1.

			Emails: gokberk.erdogan@tum.de, georg.maringer@tum.de, nikita.polyansky@gmail.com
	}}

	\IEEEoverridecommandlockouts
	\maketitle	

	\begin{abstract}
		
		In this work, we consider $q$-ary signature codes of length $k$ and size $n$ for a noisy adder multiple access channel. A signature code in this model has the property that any subset of codewords can be uniquely reconstructed based on any vector that is obtained from the sum (over integers) of these codewords. We show that there exists an algorithm to construct a signature code of length $k = \frac{2n\log{3}}{(1-2\tau)\left(\log{n} + (q-1)\log{\frac{\pi}{2}}\right)} +\mathcal{O}\left(\frac{n}{\log{n}(q+\log{n})}\right)$ capable of correcting  $\tau k$ errors at the channel output, where $0\le \tau < \frac{q-1}{2q}$. Furthermore, we present an explicit construction of signature codewords with polynomial complexity being able to correct up to $\left( \frac{q-1}{8q} - \epsilon\right)k$ errors for a codeword length $k = \mathcal{O} \left ( \frac{n}{\log \log n} \right )$, where $\epsilon$ is a small non-negative number. Moreover, we prove several non-existence results (converse bounds) for $q$-ary signature codes enabling error correction.
	\end{abstract}
	
	\section{Introduction}
	
	The problem of determining an $n$-bit number by asking questions for the sum of a subsequence of its bits has been addressed by scientists from different research backgrounds for over fifty years. In coding theory, the problem corresponds to signature codes for the \emph{binary adder multiple access channel (MAC)}, whereas in combinatorics, it is commonly referred to as the \emph{coin weighing problem}.
	
	In this paper, we discuss a non-binary generalization of this problem. Consider a set of $n$ users where the $i$-th user gets assigned a \emph{signature codeword} $\vec{s}_i \in \{0,1,\ldots,q-1\}^{k \times 1}$ with $q \geq 2$. Each user is associated to a channel input $X_i, \, i \in \{1,2,\ldots,n\}$. If a user is active, they send their signature codeword ($X_i = \vec{s}_i$); otherwise, they send the all-zero word ($X_i = \{0\}^{k \times 1}$). The output of the $q$-ary adder MAC is then defined to be $Y \triangleq \sum_{i=1}^n X_i, Y \in \{0,1,\dots,n(q-1)\}^{k \times 1}$. The goal is to identify the set of active users by observing the output $Y$, assuming each user's signature codeword is given. Thereby, the natural question is how to minimize the length $k$ of a signature code given its size $n$.
    
    For $q = 2$, an equivalent scenario is the coin weighing problem with non-adaptive weighings. Within this problem setting there are a total of $n$ coins, some of which are genuine coins with weight $g$ and the rest are counterfeit coins with weight $c$, where both $g$ and $c$ are known. The goal is to determine the counterfeit coins with as few weighing operations as possible. A weighing operation determines the weight of an arbitrary subset of the initial $n$-coin set. The term \emph{non-adaptive} refers to the fact that the weighing operations cannot be adapted according to the results of previous weighings.
    
    Representing $n$ coins, we define $\vec{u} \in \{0,1\}^{1 \times n}$ as the unknown information sequence. In addition, we define $M\in\{0,1\}^{k\times n}$ as a predefined binary query matrix, where each row is a weighing operation, i.e. a query. We finally define the answer sequence $\vec{w} \in \{0,1,\dots,n\}^{k \times 1}$ and by definition $\vec{w} \triangleq M \vec{u}^{T}$.
    
    Referring back to adder MAC, the nonzero positions in $\vec{u}$ represent the active users, and the zero positions represent the inactive users. So the coin weighing problem translates to active users sending the respective column of the query matrix through the channel, and inactive users sending the all-zero codeword. Also note that $\vec{w}$ is equivalent to the channel output $Y$.
    
    In this paper, we will mainly investigate signature codes in a combinatorial error model for a noisy $q$-ary MAC. This means that if up to $\tau k$ errors are injected in an adversarial manner into the channel output $Y$ for $0 \leq \tau \leq 1$, the set of active users shall be uniquely determined.

    \subsection{Related Works}
	The coin weighing problem has been introduced by Shapiro in \cite{sha60}. Erdős and Rényi in \cite{Erd63} provided the converse bound on the number of weighings. Lindström in \cite{lin65}, and Cantor and Mills in \cite{can66} each provided explicit constructions for the coin-weighing problem which are order-optimal as the number of coins $n\to\infty$. Bshouty in \cite{bsh12} considered the coin weighing problem in the presence of noise.
    
    The problem of constructing binary \textit{uniquely decodable} signature codes for the adder MAC has been studied by Chang and Weldon in~\cite{chang1979coding}, by van der Meulen and Vangheluwe in \cite{Meu04}. Gritsenko et al. in~\cite{grit2017} discuss the construction of binary signature codes for the noisy adder MAC when the number of active users is limited by a constant number. Under the guise of detecting sets, non-binary signature codes for the noiseless setting have been studied by Lindström in \cite{lin65}, where he provided an order-optimal code construction of size $n$ for any fixed alphabet size $q\ge 2$ and $n\to\infty$.
    
    A relevant research direction is devoted to the Mastermind game. In this game, one player conceals a $q$-ary $n$-length vector and the second players tries to  guess the vector by asking queries. Each query is an $n$-length vector and the response to the query is the Hamming distance between the vectors. For $q=2$, this problem has been shown to be equivalent to the coin weighing problem (up to at most one query) \cite{andras2004metric}. Chvatal in \cite{Chv83} initiated the study of the $q$-ary non-adaptive version of the Mastermind game. The follow-up works~\cite{kab18, nik19} established the order-optimal strategies for fixed $q$ and $n\to\infty$.

	\subsection{Our Contribution}
	In this paper, we consider $q$-ary signature codes for the noisy adder MAC. By the double counting argument, we obtain a converse bound for the $q$-ary case analogously to the binary result presented in \cite{bsh12}. We present a random construction for signature codes as well as two explicit constructions. The term \emph{explicit construction} refers to the fact that we know how to construct a signature code for the given case in polynomial time. The first explicit construction performs well if the number of errors $t$ is small and has been generated from scratch. It uses signature codes for the noiseless case that have been introduced by Lindström in \cite{lin65} and combines them with Reed-Solomon codes. The second construction is based on the previously mentioned random construction and is the method of choice for large $t$, e.g. $t=\tau k$ for a constant $0 \leq \tau \leq \frac{q-1}{2q}$. More specifically, the explicit construction is obtained by an exhaustive search to find short signature codes and by combining them with good low rate codes via Kronecker products. This enables the design of signature codes that tolerate a designated fraction of errors.
	
	For $q=2$, Bshouty presented several results in~\cite{bsh12} and some of his approaches inspired the methodology to obtain the probabilistic results for general $q$ presented in this paper. However, the generalization of approaches is not straightforward. In particular, the random construction required us to discover properties of the \emph{generalized Pascal's Triangle}.
	
	\subsection{Outline}
	
	The remainder of the paper is organized as follows. In Section~\ref{section:preliminaries}, we introduce the notation used throughout the paper alongside generalized Pascal's triangles. Furthermore, we define signature codes and we present a converse bound discovered by Bshouty in \cite{bsh12}. In Section~\ref{section:results}, we provide a random construction for a signature code tolerating $t<\frac{q-1}{2q} k$ errors in the channel output. Furthermore, we show the aforementioned two explicit constructions of signature codes for the $q$-ary MAC. Moreover, we  provide a converse bound for the $q$-ary case. Section~\ref{section:conclusion} concludes the paper.
	
	\section{Preliminaries}\label{section:preliminaries}
	
	\subsection{Notation}
    We denote the set of integers by $\mathbb{Z}$, the set of positive integers by $\mathbb{Z}^+$ and the set of negative integers by $\mathbb{Z}^-$. Similarly, we denote the set of nonnegative integers by $\mathbb{Z}^{+}_{0} = \mathbb{Z}^+ \cup \{0\}$. For a positive integer $b$, $\mathbb{Z}_0^{(b)}$ denotes the set $\{0,1,\dots,b-1\}$. Logarithms are to the base $2$ throughout this work and simply denoted by $\log$. $H_2(\cdot)$ denotes the binary entropy function. $wt(y)$ denotes the Hamming weight of the integer sequence $y$, where the Hamming weight is defined as the number of nonzero positions in a given sequence.  $\otimes$ denotes the Kronecker product of two matrices.
    
    \subsection{Signature Codes}
    A matrix $M\in\{0,1,\dots,q-1\}^{k\times n}$ is called a $q$-ary signature code capable of correcting $t$ errors if for any two distinct binary vectors $\vec{u}_1,\vec{u}_2\in\{0,1\}^{1\times n}$ and any integer vectors $\vec{e}_1,\vec{e}_2\in \mathbb{Z}^{k\times 1}$ with  $wt(\vec{e}_1)\le t$ and $wt(\vec{e}_2)\le t$, the vectors $M\vec{u}_1+\vec{e}_1$ and $M\vec{u}_2+\vec{e}_2$ are different in at least one position.

    \subsection{Pascal's Triangle}
    \label{subsec:pascalstriangle}
    Define $C_{k,n}$ as the coefficient of $x^k$ in the expression $(1+x)^n$, where $k,n \in \mathbb{Z}^{+}_{0}$ and $0 \leq k \leq n$. Note that $C_{k,n} = \binom{n}{k}$ are the binomial coefficients which can also be defined in a recursive way, i.e.
    $C_{k,n} = C_{k,n-1} + C_{k-1,n-1}$ or equivalently $\binom{n}{k} = \binom{n-1}{k} + \binom{n-1}{k-1}$ where the initial terms are $C_{0,0} = 1$, $C_{1,0} = 0$ and $C_{0,1} = 1$. Pascal's triangle can then be defined as a triangular arrangement of the coefficients $C_{k,n}$ over a range of values $n$, i.e. the coefficients of an expression of the form $(1+x)^n$. In other words, to build the triangle, we write a 1 in row zero, as $(1+x)^0 = 1$. Then, every coefficient in every row is determined by summing the 2 adjacent coefficients in the row above it, e.g. for $n \leq 6$:
    \begin{example} First seven rows of Pascal's triangle.
        \begin{center}
            \begin{tabular}{>{$n=$}l<{\hspace{1pt}}*{20}{c}}
            ~0 &1&&&&&&\\
                ~1 &1&1&&&&&\\
                ~2 &1&2&1&&&&\\
                ~3 &1&3&3&1&&&\\
                ~4 &1&4&6&4&1&&\\
                ~5 &1&5&10&10&5&1&\\
                ~6 &1&6&15&20&15&6&1
            \end{tabular}
        \end{center}
    \end{example}
    
    \subsection{Generalized Pascal's Triangle}
    \label{subsec:generalizedpascalstriangle}
    Now we are interested in the coefficients of the more general expression $(1+x+x^2+\dots+x^{q-1})^n$. For any $q$, we build a triangle with the same approach as in building Pascal's triangle - by determining every coefficient in every row by summing up the $q$ adjacent terms on top of it - which we are going to name \emph{$q$-ary Pascal's triangle}. This triangle and several of its identities have already been introduced in~\cite{Freund56} and \cite{bondarenko1993}.
        
    The triangle can be identified by a recurrence relation:
    \begin{equation}
        \label{eq:pascalstriangleconstruction}
        C^{(q)}_{k,n} = \sum_{j = 0}^{q-1} C^{(q)}_{k-j,n-1}
    \end{equation}
    where $C^{(q)}_{k,n}$ denotes the $k$th coefficient in the $n$th row of the $q$-ary Pascal's triangle and the leftmost coefficient in every row has index $k = 0$. Notice that in the expression $(1+x+x^2+\dots+x^{q-1})^n$, the powers of $x$ range from 0 to $n(q-1)$, so $0 \leq k \leq n(q-1)$, where $n \in \mathbb{Z}^+_0$. Then, for any $q$, the initial coefficients are $C^{(q)}_{0,0} = 1$ and $C^{(q)}_{k,1} = 1$ for any $k \in \{0,1,\dots,(q-1)\}$. Also note that for $q = 2$, $C^{(2)}$ depicts the triangular arrangement of binomial coefficients, i.e. $C^{(2)}_{k,n} = C_{k,n} = \binom{n}{k}$.
    
    \begin{example}
    First six rows of ternary Pascal's triangle, i.e. $q = 3$.
    \begin{center}
        \begin{tabular}{>{$n=$}l<{\hspace{1pt}}*{20}{c}}
            ~0 &&&&&1&&&&\\
            ~1 &&&&1&1&1&&&\\
            ~2 &&&1&2&3&2&1&&\\
            ~3 &&1&3&6&7&6&3&1&\\
            ~4 &1&4&10&16&19&16&10&4&1\\
            ~5 ~1&5&15&30&45&51&45&30&15&5&1
        \end{tabular}
    \end{center}
    \end{example}
    
    Note that the coefficient in row $n$ with index $k$ is indeed the coefficient of $x^k$ in the expansion of $(1+x+x^2)^n$, where every row's first coefficient is referred to with index $k=0$, every row's second coefficient is referred to with index $k=1$ and so on. Notice that the last coefficient in row $n$ has the index $k = n(q-1)$. So the $n$th row has $n(q-1)+1$ coefficients. This means that unless $n$ is odd \emph{and} $q$ is even, the $n$th row in a $q$-ary Pascal's Triangle has an odd number of coefficients. And for such pairs of $n$ and $q$, the coefficient at the center has the index $k = n(q-1)/2$ and will be called the \emph{central coefficient}.
    
    We now present a few properties of generalized Pascal's triangles.

    \begin{lemma}
    \label{lemma:trianglesumsquarelawshifted}
    For any $j \leq n$, $j \in \mathbb{Z}^+_{0}$:
    \begin{equation*}
        C^{(q)}_{(n-j)(q-1),2n} = \sum_{k = 0}^{(n-j)(q-1)}C^{(q)}_{k,n+j} \cdot C^{(q)}_{k,n-j}.
    \end{equation*}
    \end{lemma}
    \begin{proof}
    \begin{eqnarray*}
      \sum_{k = 0}^{2n(q-1)} C^{(q)}_{k,2n}x^k & = & (1+x+x^2+\dots+x^{q-1})^{2n} \\
      & = & (1+x+\dots+x^{q-1})^{n+j} \\
      & \cdot & (1+x+\dots+x^{q-1})^{n-j} \\
      & = & \left(\sum_{k_1 = 0}^{(n+j)(q-1)} C^{(q)}_{k_1,n+j}x^{k_1}\right) \\ 
      &\cdot& \left(\sum_{k_2 = 0}^{(n-j)(q-1)} C^{(q)}_{k_2,n-j}x^{k_2}\right). \nonumber
    \end{eqnarray*}
    It is known that the coefficient of $x^{(n-j)(q-1)}$ in the expansion of $(1+x+x^2+\dots+x^{q-1})^{2n}$ is $C^{(q)}_{(n-j)(q-1),2n}$. The coefficient of $x^{(n-j)(q-1)}$ on the right-hand side of the equation can be determined as:
    \begin{multline*}
            \sum_{k_1+k_2 = (n-j)(q-1)} C^{(q)}_{k_1,n+j} \cdot  C^{(q)}_{k_2,n-j} \\
         =  \sum_{k_1 = 0}^{(n-j)(q-1)} C^{(q)}_{k_1,n+j} \cdot  C^{(q)}_{(n-j)(q-1) - k_1,n-j}  \\
         =  \sum_{k_1 = 0}^{(n-j)(q-1)} C^{(q)}_{k_1,n+j} \cdot  C^{(q)}_{k_1,n-j}             
    \end{multline*}
    where the last equality is due to the symmetric structure of generalized Pascal's triangles. Hence the result in Lemma ~\ref{lemma:trianglesumsquarelawshifted} is obtained.
    \end{proof}
    \begin{example}
    Set $n = 2$, $q = 3$, $j = 1$. Then it should hold that $C^{(3)}_{2,4} = \sum_{k = 0}^{2}C^{(3)}_{k,3} \cdot C^{(3)}_{k,1}$. Setting in the coefficients we obtain $10 = 1\cdot1 + 3 \cdot 1 + 6 \cdot 1$.
    \end{example}
    
    \begin{corollary}
    \label{corollary:trianglesumsquarelaw}
    The sum of the squares of the coefficients in the $n$th row of the $q$-ary Pascal's triangle is equal to the central coefficient in the $2n$th row of the same Pascal's triangle:
    \begin{equation*}
        \sum_{k = 0}^{n(q-1)} \left(C^{(q)}_{k,n}\right)^2 = C^{(q)}_{n(q-1),2n}
    \end{equation*}
    where the central coefficient in row $2n$ has the index $k = 2n(q-1)/2 = n(q-1)$.
    \end{corollary}
    
    \begin{proof} 
    Set $j = 0$ in the proof of Lemma~\ref{lemma:trianglesumsquarelawshifted} and the proof follows.
    \end{proof}
    \begin{example}
    Consider the central coefficient on row 4, i.e. $C_{4,4}^{(3)} = 19$. Summing the squares of the coefficients on row 2 results in $1^2 + 2^2 + 3^2 + 2^2 + 1^2 = 19 = C_{4,4}^{(3)}$.
    \end{example}
    
    \begin{corollary}
    \label{corollary:fornovertwo}
    For any $j \leq n$, $j \in \mathbb{Z}^+_{0}$:
    \begin{equation*}
        \sum_{k = 0}^{n(q-1)}\left(C^{(q)}_{k,n}\right)^2 \geq \sum_{k = 0}^{(n-j)(q-1)}C^{(q)}_{k,n+j} \cdot C^{(q)}_{k,n-j}
    \end{equation*}
    with equality if and only if $j = 0$.
    \end{corollary}
    \begin{proof}
    Due to the structure of Pascal's triangle, the central coefficient is the largest coefficient in any row, and $C^{(q)}_{n(q-1),2n} \geq C^{(q)}_{k,2n}$ holds for any $0 \leq k \leq 2n(q-1)$, as $n(q-1)$ is the index of the central coefficient in row $2n$. Then the statement simply follows from Lemma~\ref{corollary:trianglesumsquarelaw} and Corollary~\ref{lemma:trianglesumsquarelawshifted}.
    \end{proof}
    
    Next we present an upper bound on a multinomial which is used later on in Lemma~\ref{lemma:upperboundforqlargerthan3}.
    \begin{lemma}
    \label{lemma:upperboundonmultinomial}
    Let $a_i \in \mathbb{Z}_0^+, \forall i \in \{1,2,\ldots,q\}$. We define $A_q := \sum_{i = 1}^{q}a_i$. Then the following bound holds:
    \begin{equation*}
        \binom{A_q}{a_1,a_2,\dots,a_q} \leq \frac{1}{(2\pi)^{(q-1)/2}} \cdot \frac{A_q^{A_q + 1/2}}{\prod_{i = 1}^{q}a_i^{a_i + 1/2}} 
    \end{equation*}
    \end{lemma}
    \begin{proof}
    Define $A_i = \sum_{j = 1}^{i}a_j$. Then, utilizing the upper bound on binomials obtained from~\cite{gal68}, we get
    \begin{eqnarray*}
        \binom{A_q}{a_1,a_2,\dots,a_q} 
        &=& \prod_{i = 2}^{q} \binom{A_i}{a_i} \\
        &\leq& \prod_{i = 2}^{q} \sqrt{\frac{A_i}{2\pi a_{i}A_{i-1}}}\cdot 2^{A_i \cdot H_2(a_i / A_i)} \\
        &=& \prod_{i = 2}^{q} \sqrt{\frac{A_i}{2\pi a_{i}A_{i-1}}}\cdot \frac{A_i^{A_i}}{a_i^{a_i} \cdot A_{i-1}^{A_{i-1}}} \\
    \end{eqnarray*}
    from which the lemma is obtained.
    \end{proof}
    
    We now present two upper bounds on the central coefficient.
    
    \begin{lemma}
    \label{lemma:upperboundontheprobabilityofthecentralcoefficient}
    For $n \geq 1$, the central coefficient in the $n$th row of a $q$-ary Pascal's triangle is upper bounded by:
    \begin{equation*}
          C^{(q)}_{n(q-1)/2,n} \leq q^{n-1}.
    \end{equation*}
    \end{lemma}
    \begin{proof}
    We start by acknowledging that the sum of the coefficients in the $n$th row of a $q$-ary Pascal's Triangle is $q^n$. Then notice Eq.~\ref{eq:pascalstriangleconstruction}. It means that every coefficient in a generalized Pascal's Triangle is the sum of $q$ coefficients on top of it. This in return means that every coefficient in a generalized Pascal's triangle contributes to $q$ coefficients in the next row. According to Eq.~\ref{eq:pascalstriangleconstruction} we have that the central coefficient of the $n$-th row in the $q$-ary Pascal's triangle is the sum of $q$ coefficients of the $n-1$-th row. Furthermore, each of these summands is contributing to $q$ coefficients within the $n$-th row. Therefore, the sum of all coefficients within the $n$-th row is lower bounded by $q C^{(q)}_{n(q-1)/2,n}$ and since the sum of all coefficients in the $n$-th row is $q^n$ the upper bound on the central coefficient follows.
    \end{proof}
    
    \begin{lemma}
    \label{lemma:upperboundforqlargerthan3}
    The central coefficient in the $n$th row of the $q$-ary Pascal's triangle is upper bounded by:
    \begin{eqnarray*}
          C^{(q)}_{n(q-1)/2,n} \leq \frac{q^{n+1}}{\sqrt{n}}\cdot c_q
    \end{eqnarray*}
    for a constant $c_q = \frac{1}{2}\left(\frac{2}{\pi}\right)^{(q-1)/2} \frac{e}{e-1}$.
    \end{lemma}
    \begin{proof}
    For parameters $n'$, $k$, and $q$, we are trying to upper bound the number of different sets consisting of $n'$ elements that sum up to $k = n'(q-1)/2$, where each element belongs to the set $\{0,1,2, \dots ,q-1\}$, i.e. we are interested in upper bounding the central coefficient in the $n'$th row of a $q$-ary Pascal's Triangle.
    
    Every coefficient in any $q$-ary Pascal's Triangle is a sum of multinomials that meet two specific conditions:
    \begin{enumerate}
        \item $\sum_{l=0}^{q-1} i_l = n'$
        \item $\sum_{l=0}^{q-1} l \cdot i_l = k$
    \end{enumerate}
    with $i_l \in \mathbb{Z}_0^+$ and therefore
    \begin{eqnarray}
    \label{eq:multinomialsumtoupperbound}
          C^{(q)}_{k,n'} = \underbrace{\sum_{i_0} \sum_{i_1} \dots \sum_{i_{q-1}}} _{\mathclap{\substack{i_0+i_1+\dots+i_{q-1} = n' \\ 0\cdot i_0+1\cdot i_1+\dots+(q-1)\cdot i_{q-1} = k}}}\binom{n'}{\underbrace{i_0, i_1, \dots, i_{q-1}}_\text{$q$ terms}}.
    \end{eqnarray}
    The sum is over all multinomials that satisfy these two conditions, where we set $k = n'(q-1)/2$. Recall that the $k$th coefficient in the $n'$th row in a $q$-ary Pascal's triangle is the coefficient of $x^k$ in the expression $(1+x+x^2+\dots+x^{q-1})^{n'}$. Also recall that this is equal to the number of ways to write down $n'$ nonnegative integers, all of which being strictly smaller than $q$, such that their sum is $k$, which is what is denoted by the two conditions above.
     
    Note that there are $q$ variables $i_{\ell}$ against two constraints, i.e. there are $q-2$ degrees of freedom. Obviously for $q = 2$, there are no degrees of freedom, as the two constraints require that $i_1 = k$ and $i_0 = n-k$, resulting in $C^{(q)}_{k,n} = \binom{n}{k}$.

    Our goal is to upper bound the sum depicted in Eq.~\eqref{eq:multinomialsumtoupperbound} for general $q$.  
    
    We set $n' = qn$ and denote the sum which we want to upper bound:
    \begin{eqnarray}
          C^{(q)}_{qn(q-1)/2,qn}
          \label{eq:sumtoupperboundforqlargerthan3}
          = \underbrace{\sum_{i_0} \sum_{i_1} \dots \sum_{i_{q-1}}} _{\mathclap{\substack{i_0+i_1+\dots+i_{q-1} = qn \\ 0\cdot i_0+1\cdot i_1+\dots+(q-1)\cdot i_{q-1} = qn(q-1)/2}}}\binom{qn}{\underbrace{i_0, i_1, \dots, i_{q-1}}_\text{$q$ terms}} \\
          = \sum_{j_1} \sum_{j_2} \dots \sum_{j_{q-2}} \binom{qn}{n-j_0, n-j_1, \dots, n-j_{q-1}}, \nonumber
    \end{eqnarray}
    \begin{eqnarray}
          j_0 & = & - \frac{(q-2)j_1 + \dots + 2j_{q-3} + 1j_{q-2}}{q-1} \nonumber \\ 
       & = &-\left[ (1-\alpha _1) j_1 + \dots + (1-\alpha _{q-2}) j_{q-2}\right], \nonumber  \\ \nonumber \\
        \label{eq:jzeroandjqminusonedefinition2}
            j_{q-1} & = & - \frac{1j_1 + 2j_2 + \dots + (q-2)j_{q-2}}{q-1} \nonumber \\ 
       & = &-\left[ \alpha _1 j_1 + \dots + \alpha _{q-2} j_{q-2}\right], \nonumber \\
       \alpha _i & = & \frac{i}{q-1} \nonumber
    \end{eqnarray}
    where $ j_i, {\ } i \in \{1,2, \dots, q-2\}$ are the degrees of freedom, and the two variables $j_0$ and $j_{q-1}$ can be determined in terms of the degrees of freedom due to the two constraints $n' = qn$ and $k = qn(q-1)/2$.
    For $n' = qn$, we can write down the multinomials that satisfy the given condition, and upper bound it as follows by Lemma~\ref{lemma:upperboundonmultinomial}:
    \begin{eqnarray}
        \binom{qn}{n-j_0, n-j_1, \dots, n-j_{q-1}} \nonumber \leq \\
         \frac{\sqrt{q}\cdot q^{qn}}{(2\pi n)^{(q-1)/2}} \prod_{i = 0}^{q-1}
        \label{eq:getridofoneovertwos2} \left(\frac{1}{\left(1 - \frac{j_i}{n}\right)^{\left(1 - \frac{j_i}{n}\right)+\frac{1}{2n}}}\right)^{n}
    \end{eqnarray}
    \begin{multline*}
            \text{max} \{(1-q)n,{\ }(1-q\frac{q-1}{2i})n\} \leq j_i < n, \\ \forall i \in \{1,2, \dots, q-2\} \nonumber
    \end{multline*}
     where the bounds on $j_i$ exist simply due to these constraints that must be satisfied:
     \begin{eqnarray*}
     n - j_{i} &\geq& 0 \\ 
     n - j_{i} &\leq& n' = qn \\
     i (n - j_{i}) &\leq& k = qn(q-1)/2 .
     \end{eqnarray*}
     
     Observe that the expression $\binom{qn}{n-j_0, n-j_1, \dots, n-j_{q-1}}$ is maximized when all $j_i = 0$, for $i \in \{0,1,\dots,q-1\}$. For each $j_i$ it holds that as $j_i$ gets farther and farther away from zero - either in the positive or in the negative direction - the expression above gets smaller and smaller. Since we have $q-2$ degrees of freedom, there are $2^{q-2}$ directions going away from the central multinomial coefficient, thus we can write the sum in Eq. ~\eqref{eq:sumtoupperboundforqlargerthan3} as $2^{q-2}$ smaller sums. Then, for each of those $2^{q-2}$ sums, we can organize its summands as groups of $n^{(q-2)/2}$, which requires us to pick $j_i$'s as integer multiples of $\sqrt{n}$. Finally, we can upper bound the sum of each of those summand groups of $n^{(q-2)/2}$ by their largest summand times $n^{(q-2)/2}$, where the largest summand is, the multinomial that is closest to the central multinomial coefficient.
     
     In other words, assume a coordinate system with $q-2$ axes, where each axis is a degree of freedom $j_i$. The maximum for the multinomial occurs when all degrees of freedom are chosen as zero. These $q-2$ axes split the space into $2^{(q-2)}$ subspaces, and we want to upper bound the sum of the multinomials in each and every subset with an infinite multivariate geometric series.
    
    For $i \in \{1,2,\dots,q-2\}$, define $0 < x_i = -\frac{j_i}{n} \leq \text{min}\{(q-1), (q\frac{q-1}{2i}-1)\}$  and rewrite the denominator of the fraction in Eq. ~\eqref{eq:getridofoneovertwos2} as:
    \begin{eqnarray}
        \label{eq:denominatorofthefraction}
        \prod_{\ell = 0}^{q-1}(1+x_{\ell})^{(1+x_{\ell}) + \frac{1}{2n}}
    \end{eqnarray}
    where $x_0 = -[ (1-\alpha _1) x_1 + (1-\alpha _2) x_2 + \dots + (1-\alpha _{q-3}) x_{q-3} + (1-\alpha _{q-2}) x_{q-2}]$ and $x_{q-1} = -[ \alpha _1 x_1 + \alpha _2 x_2 + \dots + \alpha _{q-3} x_{q-3} + \alpha _{q-2} x_{q-2}]$.
    Note that by definition of $x_0$ and $x_{q-1}$, $\sum_{\ell = 0}^{q-1} x_{\ell} = 0$. For $i \in \{1,2,\dots, q-2\}$,  define $\mathcal{X}_i = \sum _{\substack{\ell = 1 \\ l \neq i}} ^{q-2} x_{\ell}$, and note that $ x_0 + x_{q-1} - x_i = \mathcal{X}_i$. Redefining the denominator of the fraction as $f(x_i)$:
    \begin{eqnarray*}
          f(x_i) &\triangleq& (1-\alpha_i x_i - \beta _i \mathcal{X}_i)^{(1-\alpha_i x_i - \beta _i \mathcal{X}_i)+ \frac{1}{2n}} \\
          &\cdot& (1+x_i)^{(1+x_i)+ \frac{1}{2n}} \\
          & \cdot& (1-(1-\alpha_i)x_i-(1-\beta_i)\mathcal{X}_i)^{1-(1-\alpha_i)x_i+ \frac{1}{2n}} \\
          & \cdot& (1-(1-\alpha_i)x_i-(1-\beta_i)\mathcal{X}_i)^{-(1-\beta_i)\mathcal{X}_i} \cdot c
    \end{eqnarray*}
    for a constant $c = \prod_{\substack{\ell = 1 \\ \ell \neq i}}^{q-2} (1+x_{\ell})^{(1+x_{\ell})+\frac{1}{2n}}$ and for some number $0 < \beta _i < 1$. Here, note that $f(x_i)$ is in fact the same expression as in Eq.~\eqref{eq:denominatorofthefraction}, with $f(x_i)$ treating the degrees of freedom other than $x_i$ as constants.
    
    Also defining $f_{y_i}(x_i)$:
    \begin{eqnarray*}
          f_{y_i}(x_i) & \triangleq & (1-\alpha_i x_i - \beta _i \mathcal{X}_i)^{(1-\alpha_i y_i - \beta _i \mathcal{X}_i) + \frac{1}{2n}} \\
          &\cdot& (1+x_i)^{(1+y_i) + \frac{1}{2n}} \\
          & \cdot& (1-(1-\alpha_i)x_i-(1-\beta_i)\mathcal{X}_i)^{1-(1-\alpha_i)y_i + \frac{1}{2n}} \\
          & \cdot& (1-(1-\alpha_i)x_i-(1-\beta_i)\mathcal{X}_i)^{-(1-\beta_i)\mathcal{X}_i}\cdot c
    \end{eqnarray*}
    which is essentially $f(x_i)$ except that the $x_i$'s in the exponents are replaced with $y_i$'s, where $x_i < y_i = -\frac{j_i'}{n}$.
    
    The derivative of $f_{y_i}(x_i)$ is then:
    \begin{eqnarray*}
          f'_{y_i}(x_i) & = & \frac{d}{dx_i}f_{y_i}(x_i) \\
          & = & c\left(1+y_i+\frac{1}{2n}\right)\left(1+x_i\right)^{y_i+\frac{1}{2n}} \\
          &\cdot& \left(1-\left(1-{\alpha}_i\right)x_i-\left(1-{\beta}_i\right)\mathcal{X}_i\right)^{1-\left(1-{\alpha}_i\right)y_i} \\ 
          &\cdot& \left(1-\left(1-{\alpha}_i\right)x_i-\left(1-{\beta}_i\right)\mathcal{X}_i\right)^{-\left(1-{\beta}_i\right)\mathcal{X}_i+\frac{1}{2n}} \\
          &\cdot & \left(1-{\alpha}_ix_i-{\beta}_i\mathcal{X}_i\right)^{1-{\alpha}_i y_i-{\beta}_i\mathcal{X}_i+\frac{1}{2n}}  \\ 
          &-& c\left(1-\left(1-{\alpha}_i\right)y_i-\left(1-{\beta}_i\right)\mathcal{X}_i+\frac{1}{2n}\right) \\ 
          &\cdot& \left(1-{\alpha}_i\right) \left(1+x_i\right)^{1+y_i+\frac{1}{2n}}  \\ 
          &\cdot& \left(1-\left(1-{\alpha}_i\right)x_i-\left(1-{\beta}_i\right)\mathcal{X}_i\right)^{-\left(1-{\alpha}_i\right)y_i}  \\ 
          &\cdot& \left(1-\left(1-{\alpha}_i\right)x_i-\left(1-{\beta}_i\right)\mathcal{X}_i\right)^{-\left(1-{\beta}_i\right)\mathcal{X}_i+\frac{1}{2n}}  \\ 
          &\cdot& \left(1-{\alpha}_ix_i-{\beta}_i\mathcal{X}_i\right)^{1-{\alpha}_i y_i-{\beta}_i\mathcal{X}_i+\frac{1}{2n}} \\ 
          &-& {\alpha}_i c\left(1-{\alpha}_i y_i-{\beta}_i\mathcal{X}_i+\frac{1}{2n}\right) \left(1+x_i\right)^{1+y_i+\frac{1}{2n}}  \\ 
          &\cdot& \left(1-\left(1-{\alpha}_i\right)x_i-\left(1-{\beta}_i\right)\mathcal{X}_i\right)^{1-\left(1-{\alpha}_i\right)y_i}  \\ 
          &\cdot& \left(1-\left(1-{\alpha}_i\right)x_i-\left(1-{\beta}_i\right)\mathcal{X}_i\right)^{-\left(1-{\beta}_i\right)\mathcal{X}_i+\frac{1}{2n}}  \\ 
          &\cdot& \left(1-{\alpha}_ix_i-{\beta}_i\mathcal{X}_i\right)^{-{\alpha}_i y_i-{\beta}_i\mathcal{X}_i+\frac{1}{2n}} \\
          & = &  c(1+x_i)^{y_i+\frac{1}{2n}} \\ 
          &\cdot& \left(1-\left(1-{\alpha}_i\right)x_i-\left(1-{\beta}_i\right)\mathcal{X}_i\right)^{-\left(1-{\alpha}_i\right)y_i}  \\
          &\cdot& \left(1-\left(1-{\alpha}_i\right)x_i-\left(1-{\beta}_i\right)\mathcal{X}_i\right)^{-\left(1-{\beta}_i\right)\mathcal{X}_i+\frac{1}{2n}}  \\
          &\cdot& \left(1-{\alpha}_ix_i-{\beta}_i\mathcal{X}_i\right)^{-{\alpha}_i y_i-{\beta}_i\mathcal{X}_i+\frac{1}{2n}}  \\
          &\cdot \bigl[& (1+y_i+\frac{1}{2n})\left(1-\left(1-{\alpha}_i\right)x_i-\left(1-{\beta}_i\right)\mathcal{X}_i\right) \\ 
          &\cdot& \left(1-{\alpha}_ix_i-{\beta}_i\mathcal{X}_i\right)  \\ 
          &-& \left(1-\left(1-{\alpha}_i\right)y_i-\left(1-{\beta}_i\right)\mathcal{X}_i+\frac{1}{2n}\right) \\ 
          &\cdot& (1-\alpha_i)(1+x_i)\left(1-{\alpha}_ix_i-{\beta}_i\mathcal{X}_i\right)  \\
          &-& \alpha_i\left(1-{\alpha}_i y_i-{\beta}_i\mathcal{X}_i+\frac{1}{2n}\right)(1+x_i) \\ 
          &\cdot& \left(1-\left(1-{\alpha}_i\right)x_i-\left(1-{\beta}_i\right)\mathcal{X}_i\right) \bigr]
    \end{eqnarray*}
    which is always nonnegative if the condition $y_i \geq x_i + \frac{1}{2n}$ holds. Since $y_i > x_i$ and since we pick $j_i$ and $j_i'$ such that $\frac{j_i}{\sqrt{n}}$, $\frac{j_i'}{\sqrt{n}}\in \mathbb{Z}$, it already holds that $y_i \geq x_i + \frac{1}{\sqrt{n}}$, hence the condition is fulfilled.

    As a result, it holds that $f_{y_i}(y_i) \geq f_{y_i}(x_i)$, meaning that we can upper bound the fraction in Eq. ~\eqref{eq:getridofoneovertwos2} for $\text{max} \{(1-q)n,{\ }(1-q\frac{q-1}{2i})n\} \leq j_i < 0$.
    
    Similarly, define $0 < x_i = \frac{j_i}{n} < 1$ and rewrite the denominator of the fraction in Eq. ~\eqref{eq:getridofoneovertwos2} as:
    \begin{eqnarray}
        \label{eq:denominatorofthefraction2}
          \prod_{\ell = 0}^{q-1}(1-x_{\ell})^{(1-x_{\ell})+\frac{1}{2n}}
    \end{eqnarray}
    where $x_0 = -[(1-\alpha _1) x_1 + (1-\alpha _2) x_2 + \dots + (1-\alpha _{q-3}) x_{q-3} + (1-\alpha _{q-2}) x_{q-2}]$, $x_{q-1} = -[ \alpha _1 x_1 + \alpha _2 x_2 + \dots + \alpha _{q-3} x_{q-3} + \alpha _{q-2} x_{q-2}]$.
    Note that by definition of $x_0$ and $x_{q-1}$, $\sum_{\ell = 0}^{q-1} x_{\ell} = 0$. For $i \in \{1,2,\dots, q-2\}$,  define $\mathcal{X}_i = \sum _{\substack{\ell = 1 \\ l \neq i}} ^{q-2} x_{\ell}$, and note that $ x_0 + x_{q-1} - x_i = \mathcal{X}_i$. Redefining the denominator of the fraction as $g(x_i)$:
    \begin{eqnarray*}
          g(x_i) &\triangleq& (1+\alpha_i x_i + \beta _i \mathcal{X}_i)^{(1+\alpha_i x_i + \beta _i \mathcal{X}_i)+\frac{1}{2n}} \\
          &\cdot& (1-x_i)^{(1-x_i)+\frac{1}{2n}} \\ 
          &\cdot& (1+(1-\alpha_i)x_i+(1-\beta_i)\mathcal{X}_i)^{1+(1-\alpha_i)x_i+\frac{1}{2n}} \\
          &\cdot& (1+(1-\alpha_i)x_i+(1-\beta_i)\mathcal{X}_i)^{(1-\beta_i)\mathcal{X}_i} \cdot c
    \end{eqnarray*}
    for a constant $c = \prod_{\substack{\ell = 1 \\ \ell \neq i}}^{q-2} (1+x_{\ell})^{(1+x_{\ell})+\frac{1}{2n}}$ and for some number $0 < \beta _i < 1$. Here, note that $g(x_i)$ is in fact the same expression as in Eq. ~\eqref{eq:denominatorofthefraction2}, with $g(x_i)$ treating the degrees of freedom other than $x_i$ as constants.
    
    Also defining $g_{y_i}(x_i)$:
    \begin{eqnarray*}
          g_{y_i}(x_i) & \triangleq & (1+\alpha_i x_i + \beta _i \mathcal{X}_i)^{(1+\alpha_i y_i + \beta _i \mathcal{X}_i)+\frac{1}{2n}} \\ 
          &\cdot& (1-x_i)^{(1-y_i)+\frac{1}{2n}} \\
          &\cdot& (1+(1-\alpha_i)x_i+(1-\beta_i)\mathcal{X}_i)^{(1-\alpha_i)y_i+\frac{1}{2n}} \\
          &\cdot& (1+(1-\alpha_i)x_i+(1-\beta_i)\mathcal{X}_i)^{(1-\beta_i)\mathcal{X}_i}\cdot c
    \end{eqnarray*}
    which is essentially $g(x_i)$ except that the $x_i$'s in the exponents are replaced with $y_i$'s, where $x_i < y_i = \frac{j_i}{n}$.
    
    The derivative of $g_{y_i}(x_i)$ is then:
    \begin{eqnarray*}
          g'_{y_i}(x_i) & = & \frac{d}{d x_i}g_{y_i}(x_i)  \\
          & = & -c\left(1-y_i+\frac{1}{2n}\right)\left(1-x_i\right)^{-y_i+\frac{1}{2n}} \\ 
          &\cdot& \left(1+\left(1-{\alpha}_i\right)x_i+\left(1-{\beta}_i\right)\mathcal{X}_i\right)^{1+\left(1-{\alpha}_i\right)y_i}  \\
          &\cdot& \left(1+\left(1-{\alpha}_i\right)x_i+\left(1-{\beta}_i\right)\mathcal{X}_i\right)^{\left(1-{\beta}_i\right)\mathcal{X}_i+\frac{1}{2n}}  \\ 
          &\cdot& \left(1+{\alpha}_ix_i+{\beta}_i\mathcal{X}_i\right)^{1+{\alpha}_i y_i+{\beta}_i\mathcal{X}_i+\frac{1}{2n}} \\ 
          &+& \left(1-{\alpha}_i\right)c\left(1-x_i\right)^{1-y_i+\frac{1}{2n}} \\ 
          &\cdot&\left(1+\left(1-{\alpha}_i\right)y_i+\left(1-{\beta}_i\right)\mathcal{X}_i+\frac{1}{2n}\right)  \\ 
          &\cdot& \left(1+\left(1-{\alpha}_i\right)x_i+\left(1-{\beta}_i\right)\mathcal{X}_i\right)^{\left(1-{\alpha}_i\right)y_i} \\ 
          &\cdot& \left(1+\left(1-{\alpha}_i\right)x_i+\left(1-{\beta}_i\right)\mathcal{X}_i\right)^{\left(1-{\beta}_i\right)\mathcal{X}_i+\frac{1}{2n}} \\ 
          &\cdot& \left(1+{\alpha}_ix_i+{\beta}_i\mathcal{X}_i\right)^{1+{\alpha}_i y_i+{\beta}_i\mathcal{X}_i+\frac{1}{2n}}  \\ 
          &+&  c\left(1-x_i\right)^{1-y_i+\frac{1}{2n}}\left(1+{\alpha}_i y_i+{\beta}_i\mathcal{X}_i+\frac{1}{2n}\right)  \\ 
          &\cdot& {\alpha}_i \left(1+\left(1-{\alpha}_i\right)x_i+\left(1-{\beta}_i\right)\mathcal{X}_i\right)^{1+\left(1-{\alpha}_i\right)y_i} \\ 
          &\cdot& \left(1+\left(1-{\alpha}_i\right)x_i+\left(1-{\beta}_i\right)\mathcal{X}_i\right)^{\left(1-{\beta}_i\right)\mathcal{X}_i+\frac{1}{2n}} \\ 
          &\cdot& \left(1+{\alpha}_ix_i+{\beta}_i\mathcal{X}_i\right)^{{\alpha}_i y_i+{\beta}_i\mathcal{X}_i+\frac{1}{2n}}  \\
          & = & c(1-x_i)^{-y_i+\frac{1}{2n}} \\ 
          &\cdot& \left(1+\left(1-{\alpha}_i\right)x_i+\left(1-{\beta}_i\right)\mathcal{X}_i\right)^{\left(1-{\alpha}_i\right)y_i}  \\ 
          &\cdot& \left(1+\left(1-{\alpha}_i\right)x_i+\left(1-{\beta}_i\right)\mathcal{X}_i\right)^{\left(1-{\beta}_i\right)\mathcal{X}_i+\frac{1}{2n}}  \\ 
          &\cdot& \left(1+{\alpha}_ix_i+{\beta}_i\mathcal{X}_i\right)^{{\alpha}_i y_i+{\beta}_i\mathcal{X}_i+\frac{1}{2n}}  \\ 
          &\cdot \bigl[& -\left(1+\left(1-{\alpha}_i\right)x_i+\left(1-{\beta}_i\right)\mathcal{X}_i\right) \\ 
          &\cdot& (1-y_i+\frac{1}{2n}) \left(1+{\alpha}_ix_i+{\beta}_i\mathcal{X}_i\right)  \\ 
          &+& (1-\alpha_i)(1-x_i)\left(1+{\alpha}_ix_i+{\beta}_i\mathcal{X}_i\right) \\
          &\cdot& \left(1+\left(1-{\alpha}_i\right)y_i+\left(1-{\beta}_i\right)\mathcal{X}_i+\frac{1}{2n}\right)  \\
          &+& \alpha_i (1-x_i)\left(1+{\alpha}_i y_i+{\beta}_i\mathcal{X}_i+\frac{1}{2n}\right) \\ 
          &\cdot& \left(1+\left(1-{\alpha}_i\right)x_i+\left(1-{\beta}_i\right)\mathcal{X}_i\right) \bigr]
    \end{eqnarray*}
    which is always nonnegative if the condition $y_i \geq x_i + \frac{1}{2n}$ holds. Since $y_i > x_i$ and since we pick $j_i$ and $j_i'$ such that $\frac{j_i}{\sqrt{n}}$, $\frac{j_i'}{\sqrt{n}}\in \mathbb{Z}$, it already holds that $y_i \geq x_i + \frac{1}{\sqrt{n}}$, hence the condition is fulfilled.

    As a result, it holds that $g_{y_i}(y_i) \geq g_{y_i}(x_i)$, meaning that we can upper bound the fraction in Eq. ~\eqref{eq:getridofoneovertwos2} for $0 < j_i < n$. This proves that when upper bounding the fraction, for $j_i > 0$ we can replace $\left(1-\frac{j_i}{n}\right)^{\left(1-\frac{j_i}{n}\right) + \frac{1}{2n}}$ by $\left(1-\frac{1}{\sqrt{n}}\right)^{\left(1-\frac{j_i}{n}\right) + \frac{1}{2n}}$ and for $j_i < 0$ we can replace $\left(1-\frac{j_i}{n}\right)^{\left(1-\frac{j_i}{n}\right) + \frac{1}{2n}}$ by $\left(1+\frac{1}{\sqrt{n}}\right)^{\left(1-\frac{j_i}{n}\right) + \frac{1}{2n}}$, as $\frac{j_i}{\sqrt{n}} \in \mathbb{Z}$ by definition.
    
    Define $\mathcal{S} = \{1,2,\dots,q-2\}$ as the set of indices of the degrees of freedom $j_i$, i.e. $i \in \mathcal{S}$. Also define the sets $\mathcal{S_+}$, $\mathcal{S_-}$ and $\mathcal{S}_0$, denoting the set of indices $i$ with $j_i > 0$, the set of indices $i$ with $j_i < 0$, and the set of indices $i$ with $j_i = 0$ respectively.
    
    We can then upper bound the fraction in Eq.~\eqref{eq:getridofoneovertwos2} by:
    \begin{eqnarray*}
    && \prod_{i = 0}^{q-1}\frac{1}{\left(1-\frac{j_{i}}{n}\right)^{n\left(1-\frac{j_{i}}{n}\right)  + \frac{1}{2}}} \\ 
    &<& \frac{1}{\left(1+\left(\displaystyle \sum_{i \in \mathcal{S_+}}\frac{1-\alpha_i}{\sqrt{n}} - \displaystyle \sum_{i \in \mathcal{S_-}}\frac{1-\alpha_i}{\sqrt{n}}\right)\right)^{n\left(1-\frac{j_{0}}{n}\right)+\frac{1}{2}}}  \\ 
    &\cdot& \frac{1}{\left(1+\left(\displaystyle \sum_{i \in \mathcal{S_+}}\frac{\alpha_i}{\sqrt{n}} - \displaystyle \sum_{i \in \mathcal{S_-}}\frac{\alpha_i}{\sqrt{n}}\right)\right)^{n\left(1-\frac{j_{q-1}}{n}\right)+\frac{1}{2}}}  \\ 
    &\cdot& \displaystyle \prod_{i \in \mathcal{S_+}}\frac{1}{\left(1-\frac{1}{\sqrt{n}}\right)^{n\left(1-\frac{j_i}{n}\right)+\frac{1}{2}}} \\ 
    &\cdot& \displaystyle \prod_{i \in \mathcal{S_-}}\frac{1}{\left(1+\frac{1}{\sqrt{n}}\right)^{n\left(1-\frac{j_i}{n}\right)+\frac{1}{2}}}
    \end{eqnarray*}
    where we obtained the upper bound by simply replacing each and every $j_i$ in the base terms by $\frac{sign(j_i)}{\sqrt{n}}$.
    For $q-2$ degrees of freedom, this approach allows us to upper bound the whole sum of multinomials as a sum of $2^{q-2}$ geometric series, where the common ratios $r_{\ell}$, $\ell \in \{1,2,\dots,2^{q-2}\}$ of the geometric series are fairly small. Finally, we can write that:
    \begin{eqnarray*}
          &&\displaystyle \sum_{j_1} \displaystyle \sum_{j_2} \dots \displaystyle \sum_{j_{q-2}} \binom{qn}{n-j_0, n-j_1, \dots, n-j_{q-1}}  \\ 
          &<&\frac{q^{qn}\sqrt{q}}{(2\pi n)^{(q-1)/2}} \\
          &\cdot& n^{(q-2)/2} \cdot \left(\frac{1}{1-r_1}+\frac{1}{1-r_2}+\dots+\frac{1}{1-r_{2^{q-2}}}\right) \\
          &\leq& \frac{q^{qn}\sqrt{q}}{(2\pi n)^{(q-1)/2}} \cdot n^{(q-2)/2} \cdot 2^{q-2} \cdot \frac{e}{e-1} = q^{qn}\sqrt{\frac{q}{n}} \cdot c_q
    \end{eqnarray*} 
    for a constant $c_q = \frac{1}{2}\left(\frac{2}{\pi}\right)^{(q-1)/2}\frac{e}{e-1}$. The term $\frac{e}{e-1}$ originates from the fact that $\forall \ell \in \{1,2,\dots,2^{q-2}\}$, $r_{\ell} \leq \frac{1}{e}$, while any common ratio $r_{\ell}$ is of the following form:
    \begin{eqnarray*}
    r_{\ell} &=& \frac{(1+c_{0})^{n\left(1-\frac{j_0}{n}\right)+\frac{1}{2}}}{(1+c_{0})^{n\left(1-\frac{j_0}{n} + \frac{1-\alpha_i}{\sqrt{n}} \right)+\frac{1}{2}}} \\ 
    &\cdot& \frac{ \left(1-\frac{1}{\sqrt{n}}\right)^{n\left(1- \frac{j_i}{n} \right) + \frac{1}{2}} }{\left(1-\frac{1}{\sqrt{n}}\right)^{n\left(1- \frac{j_i}{n} - \frac{1}{\sqrt{n}} \right) + \frac{1}{2}}} \\
    &\cdot& \frac{(1+c_{q-1})^{n\left(1-\frac{j_{q-1}}{n}\right)+\frac{1}{2}}}{(1+c_{q-1})^{n\left(1-\frac{j_{q-1}}{n} + \frac{\alpha_i}{\sqrt{n}} \right)+\frac{1}{2}}} \\
    & = & \frac{\left(1-\frac{1}{\sqrt{n}} \right)^{\sqrt{n}}}{ (1+c_0)^{(1-\alpha_i)\sqrt{n}} (1+c_{q-1})^{\alpha_i\sqrt{n}} }
    \end{eqnarray*}
    for some constants $c_0$ and $c_{q-1}$. Then it follows that
    \begin{equation*}
        r_{\ell} \leq \left(1-\frac{1}{\sqrt{n}} \right) ^{\sqrt{n}}.
    \end{equation*}
    The term on the right hand side approaches $\frac{1}{e}$ from below as $n$ goes to infinity, so we conclude that $r_{\ell} \leq \frac{1}{e}$.
    
    Setting $n' = n$ instead of $n' = qn$ we arrive at:
    \begin{eqnarray*}
         C^{(q)}_{n(q-1)/2,n} \leq \frac{q^{n+1}}{\sqrt{n}}\cdot c_q.
    \end{eqnarray*}
    \end{proof}

    \subsection{Binary Converse Bound}
     We now present a converse bound on the code length $k$ for $q = 2$, i.e. for the case of constructing a query matrix $M\in\{0,1\}^{k\times n}$ such that $\vec{u}\in \{0,1\}^{1 \times n}$ can be uniquely determined from $\vec{w}  = M \vec{u}^{T} + \vec{e}^{T}$, where $\vec{e}$ is an unknown error vector with $\vec{e} \in \mathbb{Z}^{1 \times k}$ and $wt(\vec{e}) \leq t$, $t < k/2$.
     
    \begin{theorem}[Converse Bound]
    \label{theorem:constanterrorsqequalstwo}
    For a signature code recovering the active users in an $n$-user binary adder MAC tolerating $t$ errors it holds that
    \begin{equation*}
        k \geq \frac{2n}{\log{n}} + t\left(1 + \frac{4}{\log{n}}\right) - \mathcal{O}\left({\frac{n\log{\log{n}}}{\log^{2}n}}\right).
    \end{equation*}
    For the case that the amount of errors grows linearly in $k$, i.e. $t=\tau k$ with $0\leq \tau \leq 1$, it holds that
        \begin{equation*}
            k \geq \frac{2n}{(1-\tau)\log n} - \mathcal{O}\left({\frac{n\log{\log{n}}}{\log^{2}n}}\right).
        \end{equation*}
    \end{theorem}
    
    \begin{remark}
        Bshouty has already proven this bound in~\cite{bsh12} for a fractional amount of errors $t$, but our result applies for any $t$.
    \end{remark}
    
    \begin{proof}

        Define $\vec{v} \in \{0,1\}^{1 \times n}$ to be a random binary vector. Assume $wt(\vec{u}) = m$ where $\vec{u} \in \{0,1\}^{1 \times n}$ is the unknown information vector, so it is deterministic. As $\textrm{E}[wt(\vec{v})] = n/2$, $\textrm{E}[\vec{u} \cdot \vec{v}] = m/2$. Then, it follows that for any integer $a \in \{0,1,\dots,m\}$:
        \begin{equation*}
            \textrm{Pr}[\vec{u} \cdot \vec{v} = a] = \frac{\binom{m}{a}2^{n-m}}{2^{n}} = \frac{\binom{m}{a}}{2^{m}}
        \end{equation*}
        since in $\vec{v}$, there must be $a$ ones whose positions coincide with any $a$ of the $m$ ones in $\vec{u}$, and the values at the remaining $n-m$ positions in  $\vec{v}$ can be chosen arbitrarily.
        
        We write:
        \begin{eqnarray*}
            \Pr\left[ - \delta n \leq  \vec{u} \cdot \vec{v} - \frac{m}{2} \leq + \delta n\right]
            & \geq & 1 - 2e^{-2\delta^{2} n} \\
            & \geq & 1 - \epsilon
        \end{eqnarray*}
        where the inequality follows from Hoeffding's inequality \cite{hoe63} and $\epsilon$ is a very small nonnegative number.
        
        Now, define $\mathcal{V}_{1} \triangleq \{\vec{v} \in \{0,1\}^{1 \times n} {\ } \text{s.t.} {\ } \vec{u} \cdot \vec{v} \in [\frac{m}{2} - \delta n; \frac{m}{2} + \delta n]\}$. It holds that $|\mathcal{V}_{1}| \geq 2^n(1- \epsilon)$.
        
        Similarly, define $\mathcal{V}_{2}, \mathcal{V}_{3},\dots, \mathcal{V}_{k} \triangleq \{\vec{v} \in \{0,1\}^{1 \times n} {\ } \text{s.t.} {\ } \vec{u} \cdot \vec{v} \in [\frac{m}{2} - \delta n; \frac{m}{2} + \delta n]\}$. Also define the intersection of all these sets as $\mathcal{V}_{inter} = \bigcap^k_{i=1} \mathcal{V}_{i}$. Obviously, $|\mathcal{V}_{inter}| \geq 2^n(1- k\epsilon)$.
        
        When each of the $k$ queries in the query matrix $M$ is picked from their respective sets $\mathcal{V}_{i}$, $i \in \{1,2,\dots,k\}$, every answer to every query is going to be in the interval $[\frac{m}{2} - \delta n; \frac{m}{2} + \delta n]$, meaning that every position in the answer vector $\vec{w}$ can have $2\delta n + 1$ different values. Because the task is to uniquely detect $\vec{u}$ from $\vec{w}$, it must hold that:
        \begin{eqnarray}
            \label{eq:conditionuniquedecodability}
            (2 \delta n + 1)^k \geq 2^n(1-k\epsilon).
        \end{eqnarray}
        
        Using the Hamming ball argument on Eq~\eqref{eq:conditionuniquedecodability}, using $2 \delta n \leq 2 \delta n + 1$, and assuming $(1-k\epsilon) \approx 1$, we write:
        \begin{eqnarray}
            (2 \delta n)^k & = & 2^n \sum_{i = 0}^{t} \binom{k}{i} n^{i} \nonumber \\ 
            & \geq & 2^n \binom{k}{t}(2\delta n)^{t} \nonumber \\
            \label{eq:binomiallowerbound} 
            & \geq & 2^n \sqrt{\frac{k}{8t(k-t)}}2^{k\text{H}_{2}(t/k)}(2\delta n)^{t} \\ 
            \label{eq:entropylowerbound}
            & \geq & 2^n \frac{1}{\sqrt{8t}}2^{2t}(2\delta n)^{t}
        \end{eqnarray}
        where Eq.~\eqref{eq:binomiallowerbound} uses the lower bound in \cite{gal68} and Eq.~\eqref{eq:entropylowerbound} stems from $\text{H}_{2}(t/k) \geq 2t/k$. Moving on, we have:
        \begin{eqnarray*}
            k & \geq & \frac{n+2t - \frac{1}{2}\log{8t}}{\log{2 \delta n}} + t.
        \end{eqnarray*}
        Set $\delta n = \sqrt{n\log{n}}$ and the result follows.
    \end{proof}
    
    \section{Main Results}\label{section:results}
    
    For the case we showed the converse bound above, we now present an explicit construction of the signature code.
    
    \begin{restatable}[Explicit Construction]{theorem}{explicitconstructionbinaryconstanterrors}
    \label{theorem:explicitconstructionbinaryconstanterrors}
    There exists an explicit construction for signature codewords that can recover the vector of active users $\vec{u}$ in the presence of up to $t<k/2$ errors for code length
    \begin{eqnarray*}
    \label{eq:binaryexplicit}
        k &\leq& k_{lin} + 2t\log{n}, \\
        k_{lin} &=& \frac{2n}{\log{n}} + \mathcal{O}\left({\frac{n\log{\log{n}}}{\log^{2}n}}\right).
    \end{eqnarray*}
    \end{restatable}
    
    \begin{remark}
        This theorem also applies to general $q$, i.e. when $M\in\{0,1,\dots,q-1\}^{k\times n}$ for $q \in \mathbb{Z}^+ \setminus \{1\}$ even though $q$ does not pop up in the expression for $k$.
    \end{remark}
    
    \begin{remark}
        For $q=2$, there is still a significant gap between the converse bound (Theorem~\ref{theorem:constanterrorsqequalstwo}) and the explicit construction (Theorem~\ref{theorem:explicitconstructionbinaryconstanterrors}). We conjecture that the construction can be improved for a constant number of errors $t$ (not linearly growing in $k$ and independent from $n$ and $k$) as we have to increase the codeword length $k$ by $\log n$ for each additional error.
    \end{remark}
    
    Now we consider general $q \in \mathbb{Z}^+ \setminus \{1\}$, i.e. $M\in\{0,1,\dots,q-1\}^{k\times n}$ and $wt(\vec{e}) \leq t$, $t < k/2$.
    
    We present a converse result on the parameters of signature codes for the $q$-ary adder MAC. This result performs better in the binary case than Theorem~\ref{theorem:constanterrorsqequalstwo} if the number of errors is large, specifically if the fraction of errors is more than $1/4$.
    
    \begin{restatable}{theorem}{nonexistence}
    \label{theorem:nonexistence}
    There is no algorithm to construct a $q$-ary signature code of size $n$ and any length $k$ that is capable of correcting more than $\frac{q-1}{2q} k$ errors at the channel output.
    \end{restatable}
    
    \begin{restatable}[Random Construction]{theorem}{bshoutytheoremtwokqarym}
    \label{theorem:bshoutytheoremtwokqarym}
        There exists an algorithm to construct a signature code of size $n$ that detects the active users with code length $k$, tolerating $t$ incorrect symbols at the channel output for
        \begin{multline*}
            k > \frac{2n\log{3}}{\log{n} + (q-1)\log{\frac{\pi}{2}}} + 2t+1 \\ +\mathcal{O}\left(\frac{n}{\log{n}(q+\log{n})}\right).
        \end{multline*}
        For the case that the amount of errors grows linearly in $k$, i.e. $t=\tau k$ with $0\leq \tau \leq \frac{q-1}{2q}$, it holds that
        \begin{multline*}
            k = \frac{2n\log{3}}{(1-2\tau)\left(\log{n} + (q-1)\log{\frac{\pi}{2}}\right)} \\  +\mathcal{O}\left(\frac{n}{\log{n}(q+\log{n})}\right).
        \end{multline*}
    \end{restatable}
    
    \begin{remark}
    We point out that Theorem~\ref{theorem:bshoutytheoremtwokqarym} shows that if $q=\mathcal{O}(n)$ the number of questions can be upper bounded by a constant ($k=\mathcal{O}(1)$) as long as $t\leq \frac{q-1}{2q} k$. Furthermore, it is remarkable that the random construction in Theorem~\ref{theorem:bshoutytheoremtwokqarym} performs worse than the explicit construction in Theorem~\ref{theorem:explicitconstructionbinaryconstanterrors} for few errors. However, as the number of errors increases the random construction performs better than the aforementioned explicit construction.
    \end{remark}
    
    \begin{restatable}[Explicit Construction]{theorem}{linearerrorsqarymexplicitconstruction}
    \label{theorem:linearerrorsqarymexplicitconstruction} 
    For a small constant $\epsilon > 0$, there exists a polynomial time algorithm to construct a signature code of length
    \begin{equation*}
        k = \mathcal{O}\left( \frac{n}{\log{\log{n}}}\right)
    \end{equation*}
       for $n$ users that is able to correct up to $(\frac{q-1}{8q} - \epsilon)k$ erroneous symbols at the channel output.
    \end{restatable}    
    
    \begin{remark}
        Notice that we have presented two explicit constructions throughout this work (Theorem~\ref{theorem:explicitconstructionbinaryconstanterrors} and Theorem~\ref{theorem:linearerrorsqarymexplicitconstruction}). The first one performs better if the amount of errors is small whereas it performs poorly when $t$ is in the order of $k_{lin}$, e.g. for $k_{lin}=\mathcal{O}\left(\frac{n}{\log n}\right)$ it follows that $k=\mathcal{O}(n)$. Conversely, the second construction performs poorly for small $t$ compared to the first construction. However, for $t$ linear in $k$ the codeword length remains in $\mathcal{O}\left(\frac{n}{\log \log n}\right )$.
    \end{remark}
    
    We now restate each theorem and provide their proofs.
    
    \explicitconstructionbinaryconstanterrors*
    
    \begin{proof}
        We first obtain an $k_lin$-by-$n$ query matrix with the explicit construction presented in \cite{lin65}, which also shows that $k_{lin}$ queries are enough to uncover an $n$-bit unknown information vector. We then apply a systematic Reed-Solomon code $[n_{RS},k_{RS},d_{RS}]_q$ on every column $M^{(i)}$. As $M^{(i)} \in \{0,1\}^{k \times 1}$, $k_{RS} = k$. Then $d_{RS} = n_{RS} - k_{RS} + 1$. As $t = \frac{d-1}{2}$, $n_{RS} = k_{RS} + 2t$.
        
        To determine the order of the field $q_{RS}$ for the Reed-Solomon code, we need to observe the answer vector $\vec{w} \in \{0,1,\dots,n\}^{k \times 1}$. This hints that the encoding on columns $M^{(i)}$ must be done over a field of order $q_{RS} > n$, where $q_{RS}$ is a prime number. By the Bertrand–Chebyshev theorem \cite{Tch52}, we know that there exists a prime number $q_{RS}$ between the integers $n$ and $2n$, so $q_{RS} \approx 2n$.
        
        Since the RS-code is defined over a field of order $q_{RS}$ the parity symbols are not necessarily binary. Therefore, we convert the parity check symbols into binary numbers of length $\lceil \log q_{RS} \rceil \approx \log 2n$. In total we extended the construction presented in \cite{lin65} by $2t \log n$ parity rows to account for at most $t$ errors. The correction of up to $t$ errors can be achieved by $\mathcal{RS}$ decoding.
    \end{proof}
    
    \nonexistence*
    
    \begin{proof}
    Define $M^{(i)}$ and $M_{(i)}$ as the $i$th column and $i$th row of the matrix $M$, respectively. Let $e_i, i \in \{1,\dots,n\}$ be the standard basis vectors of $\mathbb{R}^n$. Define $w_{\ell}, \ell \in \{0,1,\dots,q-1\}$ as the number of $\ell$'s in a row vector $M_{(i)}$. Obviously, $\sum_{\ell = 0}^{q-1}w_{\ell} = n$. Similarly define $\gamma_{\ell}, \ell \in \{0,1,\dots,q-1\}$ as the fraction of $\ell$'s in a row vector $M_{(i)}$, so it follows that $w_{\ell} = \gamma_{\ell}n$.
    
    Now, as $M$ tolerates $\tau k$ incorrect answers, for every $1 \leq i_1 < i_2 \leq n$, by Lemma~1 in \cite{bsh12}, we have $wt(M(e_{i_1}-e_{i_2})) = wt(M^{(i_1)} - M^{(i_2 )}) > 2\tau k$. So we write:
    \begin{eqnarray*}
          \binom{n}{2}(2\tau k) &<& \sum_{1 \leq i_1 < i_2 \leq n} wt(M^{(i_1)} - M^{(i_2 )}) \\
          &=& \frac{1}{2}\sum_{\ell = 0}^{q-1} w_\ell(n-w_\ell) \\
          &=& \frac{1}{2} \left(n^2 - \sum_{\ell = 0}^{q-1}w_\ell^2 \right) \\
          &=& \frac{n^2}{2} \left(1 - \sum_{\ell = 0}^{q-1}\gamma_\ell^2 \right) \\
          &\leq& \frac{n^2}{2} \left(1 - q\frac{1}{q^2} \right).
    \end{eqnarray*}
    The transition from the first line to the second is because if we have a row with $w_{\ell}$ ones and $n-w_{\ell}$ zeros, the contribution of that row to the sum is $w_\ell(n-w_\ell)$, as the difference between the digit in the $i_1$th position and the digit in the $i_2$th position of that row will be nonzero only when they are different. Considering a value $0 \leq \ell \leq q-1$, there are $w_{\ell}$ many ${\ell}$'s in that row, causing $w_\ell(n-w_\ell)$ many cases where the difference is nonzero. And the factor $1/2$ is due to the sum on the right-hand side in the first line being an ordered sum.
    
    We conclude with
    \begin{eqnarray*}
        \tau n (n-1) & < & \frac{n^2}{2} \frac{q-1}{q} \\
        \tau         & < & \frac{q-1}{2q} + \mathcal{O}(1)
    \end{eqnarray*}
    and the proof is complete.
    \end{proof}
    
    \bshoutytheoremtwokqarym*
    
    \begin{proof}
    Let $M \in \{0,1, \dots, q-1\}^{k \times n}$ be the query matrix. Consider two vectors $\vec{x},\vec{y} \in \{0,1\}^{n}$ and define $\vec{z} = \vec{x}-\vec{y} \in \{-1,0,1\}^{n}$, where the number of nonzero entries in $\vec{z}$ are $w_z = w_z^+ + w_z^-$, $w_z^+$ denoting the number of entries equal to $1$ and $w_z^-$ denoting the number of entries equal to $-1$ in $\vec{z}$. We will show that Pr$[(\exists \vec{z} \in \{-1,0,1\}^{n}) {\ } wt(M \vec{z}^T) \leq 2t] < 1$ which implies the result.
    
    Define $r_{w_z^+,w_z^-} = \text{Pr}[M_{(i)} z = 0]$, where $M_{(i)}$ is the $i$th row of $M$. Notice that $r_{w_z^+,w_z^-}$ is maximized for $w_z^+ = w_z^- = w_z/2$, as has been shown in Corollary ~\ref{corollary:fornovertwo}.
    From Lemma ~\ref{lemma:upperboundforqlargerthan3}:
    \begin{eqnarray*}
        r_{w_z^+,w_z^-} = \Pr[M_{(i)} \cdot \vec{z} = 0] &\leq& \frac{\frac{q^{w_z}q}{\sqrt{w_z}}\cdot c_q}{q^{w_z}} \\
        & = & \frac{q}{\sqrt{w_z}}\cdot c_q
    \end{eqnarray*}
    because as has been shown in Lemma ~\ref{corollary:trianglesumsquarelaw}, the sum of the squares of the elements on row $w_z/2$ in $q$-ary Pascal's Triangle is equal to the central coefficient on row $w_z$. We also know that $r_{w_z^+,w_z^-} \leq 1/q$ for $w_z \geq 1$ due to Lemma ~\ref{lemma:upperboundontheprobabilityofthecentralcoefficient}.
    
    Now we write:
    \begin{equation*}
        \text{Pr}[wt(M \vec{z}^T) \leq 2t] =
        \displaystyle \sum_{j = 0}^{2t} \left(1-r_{w_z^+,w_z^-}\right)^{j} r_{w_z^+,w_z^-}^{k-j} \binom{k}{j}
    \end{equation*}
    where $j$ denotes the number of nonzero positions in $M\vec{z}^T$.
    
    Making use of $r_{w_z^+,w_z^-} \leq 1/q$:
    \begin{eqnarray*}
        \left(1-r_{w_z^+,w_z^-}\right)^{j} r_{w_z^+,w_z^-}^{k-j} 
        &\leq& \left(1-r_{w_z^+,w_z^-}\right)^{2t} r_{w_z^+,w_z^-}^{k-2t} \\
        & = & \frac{(q-1)^{2t}}{q^k}.
    \end{eqnarray*}
    Continuing:
    \begin{eqnarray*}
        \Pr[wt(M\vec{z}^T) \leq 2t] &\leq& 
        \displaystyle \sum_{j = 0}^{2t} \frac{(q-1)^{2t}}{q^k}  \binom{k}{j} \\
        & \leq & \frac{(q-1)^{2t}}{q^k} \cdot 2^{H_2(2\tau)k}.
    \end{eqnarray*}
    Also, as $\left(1-r_{w_z^+,w_z^-}\right)^{2t-j} r_{w_z^+,w_z^-}^{k-2t+j} \leq r_{w_z^+,w_z^-}^{k-2t}$ we can state that:
    \begin{eqnarray*}
        \Pr[wt(M \vec{z}^T) \leq 2t] &\leq& r_{w_z^+,w_z^-}^{k-2t} \displaystyle \sum_{j = 0}^{2t} \binom{k}{j} \\
        & \leq & r_{w_z^+,w_z^-}^{(1-2\tau)k}\cdot2^{H_2(2\tau)k} \\
        & \leq & \left(\frac{q^2 c_q^2}{w_z}\right)^{\frac{1-2\tau}{2}k}\cdot2^{H_2(2\tau)k}.
    \end{eqnarray*}
    We can now write:
    \begin{multline*}
        \Pr[(\exists z) wt(M \vec{z}^T) \leq 2t] \\ 
        \leq \displaystyle \sum_{w_z = 1}^{n} \binom{n}{w_z} 2^{w_z} \cdot 2^{H_2(2\tau)k} \\
        \cdot \text{min} \left\{\frac{(q-1)^{2t}}{q^k}, {\ } \left(\frac{q^2 c_q^2}{w_z}\right)^{\frac{1-2\tau}{2}k} \right\} .
    \end{multline*}
    Here, note that we used the upper bound $\sum_{j = 0}^{2t} \binom{k}{j} \leq 2^{H_2(2\tau)k}$. This upepr bound applies when we consider a linear number of errors, i.e. $t = \tau k$. For a constant $t$ number of errors, the upper bound $\sum_{j = 0}^{2t} \binom{k}{j} \leq \binom{k}{2t}\cdot (2t+1)$ is used instead.
    
    Moving on,
    \begin{eqnarray*}
         && \binom{n}{w_z} 2^{w_z} \cdot \frac{(q-1)^{2t}}{q^k} \cdot 2^{H_2(2\tau)k} \\
         & = & 2^{w_z + w_z\log{n} -k\log{q} + 2t\log{(q-1)} +kH_2(2\tau)} \\
         & = & 2^{w_z(1 + \log{n}) - k\log{(q)}\left(1 - \frac{q-1}{q}\log_q{(q-1)} - \frac{H_2(2\tau)}{\log{q}}\right)} \\
         & < & \frac{1}{n}
    \end{eqnarray*}
    holds for $w_z < \frac{n}{\log^3{n}}$ as $k \approx \frac{n}{\log{n}}$, since
    \begin{eqnarray*}
    && \frac{q-1}{q}\log_q{(q-1)} + \frac{H_2(2\tau)}{\log{q}} \\ 
    & = & \frac{q-1}{q}\log_q{(q-1)} - 2\tau \log_q{2\tau} \\
    &&- (1-2\tau)\log_q{(1-2\tau)} \\
    &<& \frac{q-1}{q}\log_q{(q-1)} - \frac{q-1}{q} \log_q{\left(\frac{q-1}{q}\right)} \\ 
    && - \frac{1}{q}\log_q{\left(\frac{1}{q}\right)} = 1
    \end{eqnarray*}
    holds with inequality due to $\tau < \frac{q-1}{2q}$.
    
    For $w_z \geq \frac{n}{\log^3{n}}$:
    \begin{eqnarray}
    \label{eq:replacelinearwithconstant}
          &&\binom{n}{w_z} 2^{w_z} 2^{H_2(2\tau)k} \left(\frac{q^2 c_q^2}{w_z}\right)^{\frac{1-2\tau}{2}k} \\ \nonumber
          &\leq& 3^n 2^{H_2(2\tau)k} \left(\frac{q^{2}c_q^2\log^3{n}}{n}\right)^{\frac{1-2\tau}{2}k} \\ \nonumber
          &\leq& 2^{n\log{3}-\frac{1-2\tau}{2}k\log{n} + H_2(2\tau)k + \frac{1-2\tau}{2}k\log{\left(q^{2}c_q^2\log^3{n}\right)}} \\ \nonumber
          & < & \frac{1}{n}
    \end{eqnarray}
    must hold. Then
    \begin{multline*}
        n\log{3}-\frac{1-2\tau}{2}k\left(\log{n} - \log{\left(q^{2}c_q^2\log^3{n}\right)}\right) \\ 
        + kH_2(2\tau) < 0.
    \end{multline*}
    Notice as $k$ is on the order of $\frac{n}{\log{n}}$, the significant summands are only the first two. So we obtain $k$ as: 
    \begin{equation*}
        k > \frac{2n\log{3}+\mathcal{O}\left(\frac{n}{\log{n}}\right)}{(1-2\tau)\left(\log{n} - \log{\left(q^{2}c_q^2\log^3{n}\right)}\right)}.
    \end{equation*}
    Setting in $c_q = \frac{1}{2}\left(\frac{2}{\pi}\right)^{(q-1)/2}\epsilon$ we obtain:
    \begin{multline*}
        k > \frac{2n\log{3}}{(1-2\tau)\left(\log{n} + (q-1)\log{\frac{\pi}{2}}\right)} \\  +\mathcal{O}\left(\frac{n}{\log{n}(q+\log{n})}\right).
    \end{multline*}
    
    This is the result for linear number of errors, i.e. for $t = \tau k$. We now write $\binom{k}{2t}\cdot (2t+1)$ instead of $2^{H_2(2\tau)k}$ in Eq.~\eqref{eq:replacelinearwithconstant} to obtain the result for constant number of errors. Performing the replacement we get:
    
    \begin{eqnarray*}
          &&\binom{n}{w_z} 2^{w_z} \binom{k}{2t}\cdot (2t+1) \left(\frac{q^2 c_q^2}{w_z}\right)^{\frac{k-2t}{2}} \\
          &\leq& 3^n \sqrt{\frac{k}{2\pi2t(k-2t)}} \\ 
          && \cdot2^{kH_2(\frac{2t}{k})}\cdot (2t+1) \cdot \left(\frac{q^{2}c_q^2\log^3{n}}{n}\right)^{\frac{k-2t}{2}} \\
          &\leq& 2^{n\log{3}-\frac{k-2t}{2}\log{n} + \log{(2t+1)} + kH_2(\frac{2t}{k})} \\
          && \cdot 2^{ + \frac{k-2t}{2}\log{\left(q^{2}c_q^2\log^3{n}\right)}+\frac{1}{2}\log{\left(\frac{k}{2\pi2t(k-2t)}\right)}} <  \frac{1}{n}
    \end{eqnarray*}
    must hold. So,
    \begin{multline*}
          n\log{3}-\frac{k-2t}{2}\left(\log{n} - \log{\left(q^{2}c_q^2\log^3{n}\right)} \right) + \\ \log{(2t+1)} + kH_2(\frac{2t}{k}) +\frac{1}{2}\log{\left(\frac{k}{2\pi2t(k-2t)}\right)} < 0.
    \end{multline*} 
    Notice as $k$ is on the order of $\frac{n}{\log{n}}$ and $t \leq k/2$, the significant summands are only the first two. Then we obtain $k$ as:
    \begin{equation*}
        k > \frac{2n\log{3} + \mathcal{O}\left(\frac{n}{\log{n}}\right)}{\log{n} - \log{\left(q^{2}c_q^2\log^3{n}\right)}} + 2t+1.
    \end{equation*}
    Setting in $c_q = \frac{1}{2}\left(\frac{2}{\pi}\right)^{(q-1)/2}\epsilon$ we obtain:
    \begin{multline*}
        k > \frac{2n\log{3}}{\log{n} + (q-1)\log{\frac{\pi}{2}}} + 2t+1 \\ +\mathcal{O}\left(\frac{n}{\log{n}(q+\log{n})}\right).
    \end{multline*}
    and the proof is complete.
    \end{proof}
    
    \linearerrorsqarymexplicitconstruction*
    
    \begin{proof}
    Assume there exists two small constants $\epsilon_1$ and $\epsilon_2$ so that $(\frac{q-1}{2q} - \epsilon_1)(1/4 - \epsilon_2/2) = (\frac{q-1}{8q} - \epsilon)$. Due to Theorem~\ref{theorem:bshoutytheoremtwokqarym}, there exists a query matrix $M \in \{0,1,\dots,q-1\}^{p \times s}$ that tolerates $(\frac{q-1}{2q} - \epsilon_1)p$ errors. Now say that:
    \begin{equation*}
        p = \frac{(8\log{3})\sqrt{\log{n}}}{(\log{\log{n}})\sqrt{\log{q}}}, \:\:\:\:\: s = \frac{\sqrt{\log{n}}}{\sqrt{\log{q}}}.
    \end{equation*}
    There are at most $q^{ps} = n^{\frac{(8\log{3})}{\log{\log{n}}}} \leq n$ such matrices, and by Lemma~1 in \cite{bsh12}, a query matrix $M$ can be found in polynomial time, such that $\vec{x}$ can be recovered from $M\vec{x}^T+\vec{e}$ where $\vec{x} \in \{0,1\}^{1 \times s}$ and $wt(\vec{e}) \leq (\frac{q-1}{2q} - \epsilon_1)p$, because $\vec{x}$ and $\vec{e}$ can be found through exhaustive search. We denote this algorithm as $\mathcal{ES}$. 
     
    Now, define $r = n/s$ and define $\mathcal{C}$ as a linear code $[N, K, D] := [c_1r, r, (1/2 - \epsilon_2)c_1r]$ over $\mathbb{Z}_0^{(2)}$ with $G \in \{0,1\}^{r\times(c_1r)}$ and polynomial time decoding algorithm $\mathcal{D}$ where $c_1$ is a constant with $c_1 \geq \frac{r-1}{r}\cdot \frac{1}{\frac{1}{2}+\epsilon_2} \approx \frac{2}{1 + \epsilon_2}$ due to Singleton bound \cite{sin64}. An example to such a code are concatenation codes. Let 
    \begin{equation*}
        B \triangleq \bar{G}^T \otimes M = 
        \begin{pmatrix}
        g_{1,1}M & g_{1,2}M & \dots  & g_{1,r}M \\
        g_{2,1}M & g_{2,2}M & \dots  & g_{2,r}M \\
        \vdots   & \vdots   & \ddots & \vdots   \\
        g_{c_1r,1}M & g_{c_1r,2}M & \dots  & g_{c_1r,r}M
        \end{pmatrix}
    \end{equation*}
    where $\bar{G}$ is essentially $G$ but over $\mathbb{Z}$. Note that $B$ is a $k$-by-$n$, $n = rs$ matrix where 
    \begin{equation*}
        k = \frac{(8c_1\log{3})n}{\log{\log{n}}}.
    \end{equation*}
     
    Now we argue that $B$ tolerates $k' = (\frac{q-1}{2q} - \epsilon_1)(1/4 - \epsilon_2/2)k = (\frac{q-1}{8q} - \epsilon)k$ incorrect answers and given $ \vec{b} = B\vec{v} + \vec{e}$, where $\vec{v} \in \{0,1\}^n$, $\vec{e} \in \mathbb{Z}^k$ and $wt(\vec{e}) \leq k'$, $\vec{v}$ can be reconstructed in polynomial time.
     
    We start by dividing $\vec{v}$ into size $s$-vectors where $\vec{v}^{(\ell)} \in \{0,1\}^s$ and then define $\vec{w}^{(i)}$ as:
    \begin{equation*}
        \vec{v} = 
        \begin{pmatrix}
        \vec{v}^{(1)} \\
        \hline
        \vec{v}^{(2)} \\
        \hline
        \vdots \\
        \hline
        \vec{v}^{(r)}
        \end{pmatrix}, \:\:\:\:\:
        \vec{w}^{(i)} = 
        \begin{pmatrix}
        M_{(i)} \cdot \vec{v}^{(1)} \\
        \hline
        M_{(i)} \cdot \vec{v}^{(2)} \\
        \hline
        \vdots \\
        \hline
        M_{(i)} \cdot \vec{v}^{(r)}
        \end{pmatrix}
    \end{equation*}
    
    where $M_{(i)}$ is the $i$th row of $M$, $i \in \{1,2,\dots,p\}$. Then for 
    $\vec{b}^{(i)} = (b_{i}, b_{p+i}, b_{2p+i}, \dots, b_{(c_1r-1)p+i})^T$ and
    $\vec{e}^{(i)} = (e_{i}, e_{p+i}, e_{2p+i}, \dots, e_{(c_1r-1)p+i})^T$ we have $\vec{b}^{(i)} = \bar{G}^T\vec{w}^{(i)} + \vec{e}^{(i)}$.

    As $\vec{w}^{(i)} \in \{0,1,\dots,s(q-1)\}^r$, by Lemma~2 in \cite{bsh12}, there is a polynomial time algorithm so that if $wt(\vec{e}^{(i)}) < (1/4 - \epsilon_2/2)c_1r$, then the algorithm returns $\vec{w}^{(i)}$. Otherwise, the algorithm returns some $\vec{w}'^{(i)} \in \{0,1,\dots, 2s(q-1)\}^r$. Let $\mathcal{J} := \{j_1,j_2,\dots,j_m\} \subseteq \{1,2,\dots,p\}$ such that  $wt(\vec{e}^{(i)}) > (1/4 - \epsilon_2/2)c_1r$ for $i \in \mathcal{J}$ and $wt(\vec{e}^{(i)}) \leq (1/4 - \epsilon_2/2)c_1r$ for $i \not \in \mathcal{J}$. As 
    \begin{equation*}
        \sum_{i = 1}^p wt(\vec{e}^{(i)}) = wt(\vec{e}) \leq k'
    \end{equation*}
    it holds that
    \begin{equation*}
        m \leq \frac{k'}{(1/4 - \epsilon_2/2)c_1r} = \left(\frac{q-1}{2q}-\epsilon_1\right)p.
    \end{equation*}
    We run the algorithm for each $\vec{b}^{(i)}$ and obtain some vector $\vec{z}^{(i)}$. Due to Lemma~2 in \cite{bsh12}, $\vec{z}^{(i)}$ is some vector  $\vec{w}'^{(i)} \in \{0,1,\dots, 2s(q-1)\}^r$ if $i \in \mathcal{J}$ and $\vec{z}^{(i)} = \vec{w}^{(i)}$ if $i \not \in \mathcal{J}$. Let $\vec{a}^{(j)} = (z_j^{(1)}, \dots, z_j^{(p)})^T$ and $\vec{g}^{(j)} = (w_j^{(1)} - z_j^{(1)}, \dots, w_j^{(p)} - z_j^{(p)})^T$ for $j = \{1,2, \dots,r\}$. Since $\vec{a}^{(j)} + \vec{g}^{(j)} = M\vec{v}^{(j)}$ we have $\vec{a}^{(j)} = M\vec{v}^{(j)} - \vec{g}^{(j)}$. Since for every $j$ we have $wt(\vec{g}^{(j)}) \leq  m \leq ((q-1)/2q - \epsilon_1)p$ the vectors $\vec{v}^{(j)}$, $j = \{1,2,\dots,r\}$ can be found in polynomial time by the algorithm $\mathcal{ES}$.
    \end{proof}
    
	\section{Conclusion and Further Research Directions}\label{section:conclusion}
    In this work we covered the problem of constructing signature codes for a noisy adder MAC in the adversarial setting.
    We have shown the existence of a set of signature codewords for the adder MAC that allow perfect detection of the active users even in the event of errors at the channel output. Specifically the paper covers the case of a fixed number of errors $t$, independent of the amount of the signature codeword length $k$ and the number of codewords $n$ as well as the case that the upper bound on the number of errors is linear in $k$, i.e. $t=\tau k$.
    
    Apart from the existence of signature codes we furthermore presented two explicit constructions that can be run in polynomial time with respect to $k,n$ and $t$. We are introducing two constructions because the first one performs better for small values of $t$ while the latter one performs better if for large $t$ (e.g. if $t$ is linear in $k$).
    
    Furthermore, we show a converse result, namely that it is impossible to construct signature codes for the $q$-ary adder MAC if the number of erroneous symbols at the channel output exceeds $\frac{q-1}{2q}$.
    
    Our constructions and existence results do not match the converse bounds we present throughout this work. We conjecture that especially the first explicit construction can be improved since in the converse bound (Theorem~\ref{theorem:constanterrorsqequalstwo}) is only increased by $1$ for each error while the explicit construction is increased by $\log n$ per additional error.
    
    An additional scenario for future work would be to consider the case that the information vector $\vec{u}$ is \emph{nonbinary}, i.e. $\vec{u}\in \{0,1,\ldots,r-1\}^{1\times n}$ for $r>2, r\in \mathbb{Z}^+$. This case has been considered by Lindström in \cite{lin65} for the error-free case. In the noisy adder MAC problem, this would translate to each user having not one, but $r-1$ messages excluding the all-zero vector. In other words, each user in the channel is given a choice $r_i \in  \{0,1,\dots,r-1\}$, with which their signature codeword $\vec{s}_i \in \{0,1,\dots,q-1\}^{k \times 1}$ will be multiplied so that $\sum_{i = 1}^{n}{r_{i}\vec{s}_i} = Y$, $Y \in \{0,1,\dots,n(q-1)(r-1)\}^{k \times 1}$. Then the goal would be to determine the vector $(r_1,r_2,\dots,r_n)$ from $Y$ in the presence of $t$ errors.

	\bibliographystyle{ieeetransa}
	\bibliography{literature}
\end{document}